\documentclass[journal]{IEEEtran}
\usepackage{csquotes}
\usepackage{amsmath}
\usepackage{amsmath, amssymb, amsthm, bm}
\usepackage{csquotes}
\usepackage{amsmath}
\DeclareMathOperator*{\argmax}{arg\,max}

\newtheorem{corollary}{Corollary}
\newtheorem{proposition}{Proposition}
\newtheorem{remark}{Remark}
\usepackage{authblk}
\usepackage{hyperref}
\usepackage{float}
\usepackage{cite}
\usepackage{algorithm}
\usepackage{algorithmicx, algpseudocode}
\usepackage{multirow}

\usepackage{mathtools}
\ifCLASSINFOpdf
\usepackage[pdftex]{graphicx}
\else
\fi
\usepackage{amsmath}
\usepackage{amsfonts}
\usepackage{amssymb}
\usepackage{color}
\usepackage{array}
\usepackage{fixltx2e}
\usepackage{stfloats}
\usepackage{url}
\usepackage{amsmath, amsthm, amssymb}

\usepackage{booktabs}
\usepackage{varwidth}
\usepackage[table,xcdraw]{xcolor}

\usepackage{color,soul}

\usepackage{caption}
\usepackage{subcaption}
\usepackage[font=scriptsize]{caption}

\begin{document}
	
\title{\linespread{1} Novel Selection Schemes for Multi-RIS-Assisted Fluid Antenna Systems}
\author{Mahmoud Aldababsa,~\IEEEmembership{Senior Member,~IEEE,}
Taissir Y. Elganimi,~\IEEEmembership{Senior Member,~IEEE,}
Mahmoud~A.~Albreem,~\IEEEmembership{Senior Member,~IEEE,}
and Saeed~Abdallah,~\IEEEmembership{Member,~IEEE} 
	 
}
\maketitle
\begin{abstract}
This paper investigates the performance of a multi-reconfigurable intelligent surface (RIS)-assisted fluid antenna system (FAS). In this system, a single-antenna transmitter communicates with a receiver equipped with a planar FAS through multiple RISs in the absence of a direct link. To enhance the system performance, we propose two novel selection schemes: \textit{Max-Max} and \textit{Max-Sum}. In particular, the \textit{Max-Max} scheme selects the best combination of a single RIS and a single fluid antenna (FA) port that offers the maximum signal-to-noise ratio (SNR) at the receiver. On the other hand, the \textit{Max-Sum} scheme selects one RIS while activating all FA ports providing the highest overall SNR. We conduct a detailed performance analysis of the proposed system under Nakagami-$m$ fading channels. First, we derive the cumulative distribution function (CDF) of the SNR for both selection schemes. The derived CDF is then used to obtain approximate theoretical expressions for the outage probability (OP) and the delay outage rate (DOR). Next, a high-SNR asymptotic analysis is carried out to provide further insights into the system performance in terms of diversity and coding gains. Finally, the analytical results are validated through extensive Monte Carlo simulations, demonstrating their accuracy and providing a comprehensive understanding of the system's performance.

\end{abstract}
	
\IEEEpeerreviewmaketitle

\begin{IEEEkeywords}
Multi-reconfigurable intelligent surface, fluid antenna
system, novel selection schemes, outage probability, delay outage rate.
\end{IEEEkeywords}

	
\section{Introduction}
\IEEEPARstart{I}{n} the sixth-generation (6G) of wireless networks, the peak data rate is anticipated to reach a thousand Gbps, with an envisioned connectivity density of up to ten million users per square kilometer. Moreover, energy and spectrum efficiencies are expected to achieve ten-hundred times the current levels \cite{6G}. To date, fifth-generation (5G) technologies have incorporated massive multiple-input multiple-output (MIMO), widely considered as the most important breakthrough in the recent history of mobile technology \cite{MIMO}. The massive MIMO technology utilizes a large number of antennas to allow simultaneous transmission to multiple users within the same time-frequency resource block. 
Although it can greatly improve the spectral efficiency, it also incurs higher costs due to increased hardware requirements and computational complexity \cite{8804165}.
Due to these limitations, 5G systems may not be able to accommodate the rapidly growing demands of future wireless networks. This necessitates further evolution and development of mobile communication technologies such as reconfigurable intelligent surfaces (RISs) \cite{RIS} and fluid antenna systems (FASs) \cite{FAS}.
	
RISs have been proposed as a novel and revolutionary technology that is both affordable and energy efficient. A RIS consists of a large number of reflecting elements embedded in a planar surface, which can be controlled via software loaded to an attached smart controller. Each element in the RIS can tune the reflected signal's phase. Therefore, a considerable improvement can be achieved in the quality of the wireless channel between the transmitter and the receiver\cite{RIS_idea}. However, FASs have also been proposed as a promising solution to enhance spatial diversity and reduce hardware costs. Briefly, a fluid antenna (FA), also known as a movable antenna, is a software-controlled fluidic or dielectric structure capable of dynamically altering its position, shape, size, or radiation properties \cite{FAS_idea}. This technology is now feasible, thanks to recent advancements in antenna design via liquid metals and radio frequency (RF)-switchable pixels, which introduce new degrees of freedom that can be exploited to enhance the overall system performance. 
	
Compared to the traditional massive MIMO systems with fixed-position antennas, the RIS can be implemented at a much lower hardware cost and power consumption, since the cost per each reflecting element of the RIS is much lower than the cost per each antenna element in MIMO systems. The FAS, on the other hand, can fully leverage spatial diversity by dynamically adjusting the FA position within a defined finite region at the transceiver. Considering the impact of RIS and FAS, these two innovative technologies offer significantly more efficient alternatives to traditional massive MIMO systems \cite{RIS_FA_1}. Although RIS and FAS are inherently complementary technologies, their mutual synergy has not been adequately explored. To date, the existing literature on the combined implementation of RIS and FAS is notably sparse. In the following subsection, the most relevant work related to FAS and RIS-assisted FAS is highlighted.

\subsection{Related Works}
\subsubsection{FAS}
In recent years, numerous studies have investigated different aspects of FAS in wireless communication systems. The authors in \cite{9715064} addressed the challenge of port selection in the FAS by proposing fast algorithms that combine machine learning and analytical approximations to optimize the signal-to-noise ratio (SNR). The work in \cite{10103838} advanced FAS modeling by using a two-stage approximation, which simplifies the outage probability (OP) expression to a single integral. This approach improves channel distribution accuracy while accounting for spatial correlation under Jake’s model. In \cite{10253941}, the authors investigated the FAS under arbitrary fading with correlated channels, using copula theory to model dependencies and derive a closed-form OP expression. They further analyzed the impact of Nakagami-\emph{m} fading with Archimedean copulas to assess spatial correlation effects. The authors in \cite{10678877} modeled spatial correlation in the FAS using elliptical copulas, representing Jake’s model through a Gaussian copula. They also derived key performance metrics, including the OP and the delay outage rate (DOR), for arbitrary and Nakagami-\emph{m} correlated fading. The work in \cite{10319727} studied the OP of the FAS-equipped receiver in a point-to-point system, using Student-\emph{t} and Gaussian copula to model the spatial correlation. In \cite{10623405}, the authors proposed a block-correlation model to approximate the spatial correlation in the FAS, enabling tractable performance analysis while closely matching realistic models like the Jake’s and Clarke’s models. In \cite{10375698}, the authors examined the FAS with multiple active ports, selecting a specified subset of ports for maximum ratio combining. 

\subsubsection{RIS-assited FAS}
A small number of studies have investigated the integration of RISs and FASs. In \cite{RIS_FA_1}, the authors introduced a novel paradigm that employs RISs as distributed artificial scattering surfaces to create a rich scattering environment and enable the FAS to mitigate multiuser interference. The authors in \cite{RIS_FA_2} presented a low-complexity joint beamforming strategy tailored for the RIS-FAS, which operates solely on statistical channel state information (CSI). In \cite{RIS_FA_3}, a novel FA-empowered joint transmit and receive index modulation (IM) transmission scheme was designed for RIS-assisted millimeter-wave (mmWave) single-input multiple-output communication systems. In \cite{RIS_FA_4}, the authors derived approximate expressions for the system's OP via the central limit theorem (CLT), and used the block correlation channel model to simplify the derived OP expressions and reduce the computational complexity. In \cite{RIS_FA_5}, the authors analyzed the spatial correlation structures between the ports of the planar FAS. The joint distribution of the equivalent channel gain at the user was derived using the CLT to obtain the OP and DOR. In a recent work \cite{10826703}, the authors studied the performance of multi-RIS-aided FAS, assuming that FAS includes one RF chain, i.e., only one port is activated for communication. More recently, the OP is derived in \cite{10858773} based on the block-diagonal matrix approximation (BDMA) model to analyze the FAS-RIS system.

\subsection{Motivations and Contributions}
The aforementioned studies focus primarily on single-RIS configurations \cite{RIS_FA_1, RIS_FA_2, RIS_FA_3, RIS_FA_4, RIS_FA_5, 10858773} or multi-RISs with the activation of a single port in FAS \cite{10826703}. Although these approaches simplify system design and reduce computational complexity, they do not fully leverage the potential benefits of spatial diversity and performance enhancement that can be achieved by deploying multiple RISs or the simultaneous activation of all FAS ports. This motivates us to investigate the integration of multiple RISs and FAS. The main contributions of this paper can be summarized as follows:
\begin{itemize}
	\item We consider a multi-RIS-assisted FAS, where a single-antenna transmitter communicates with a receiver equipped with a planar FAS, facilitated by multiple RISs in the absence of a direct link between the transmitter and the receiver.
	\item We introduce two new selection schemes, denoted as \textit{Max-Max} and \textit{Max-Sum}. The \textit{Max-Max} scheme focuses on selecting the optimal combination of a single RIS and port to maximize the SNR in the receiver. Specifically, one RIS is chosen from a set of available RISs, and one port is selected from multiple available ports to achieve the highest received SNR at the receiver. On the other hand, the \textit{Max-Sum} scheme involves selecting one RIS from the available options while activating all ports to combine the signals of the chosen RIS at the receiver, thereby maximizing the overall received SNR through signal aggregation.
	\item We perform a performance analysis of multi-RIS-assisted FAS for both the \textit{Max-Max} and \textit{Max-Sum} selection schemes over Nakagami-$m$ fading channels. Initially, we derive the cumulative distribution function (CDF) of the SNR for each selection scheme. Then, approximate closed-form expressions for the OP and DOR are obtained. Next, we derive asymptotic OP expressions in the high-SNR regions. This explicitly provides deeper insights into system performance by presenting the overall performance in terms of diversity and coding gains.
	\item We validate the derived analytical results through Monte Carlo simulations, ensuring their accuracy by comparing the theoretical findings with simulation-based outcomes. 
    The results show a considerable advantage in the \textit{Max-Sum} selection scheme over the \textit{Max-Max} selection scheme. Furthermore, the DOR performance degrades as the amount of transmitted data increases, while it improves as the channel bandwidth increases. 
	\end{itemize}

\subsection{Paper Organization}
The remainder of this paper is divided into five sections as follows. Section \ref{System Model} presents the system model. Section \ref{Selection Scheme} introduces the algorithms, formal definitions, and corresponding mathematical formulations of the selection schemes. The performance analysis of the proposed system is presented in Section \ref{Performance Analysis}, while Section \ref{Numerical Analysis} offers detailed numerical evaluations and insights. Finally, Section \ref{Conclusion} summarizes the key findings and concludes the paper.




\begin{figure}[t]
\centering
\subfloat[\label{AAA}]{\includegraphics[width=7.7cm,height=4.3cm]{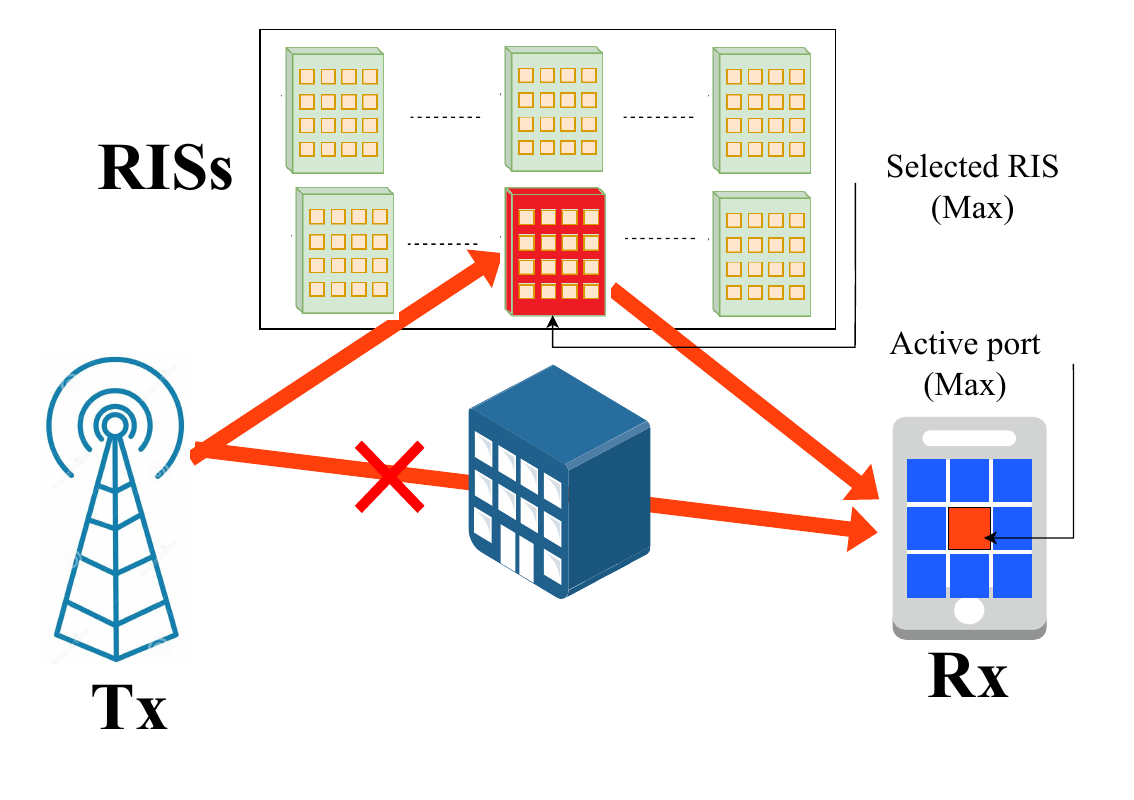}}
\hspace{\fill}
\subfloat[\label{BBB}]
{\includegraphics[width=7.7cm,height=4.3cm]{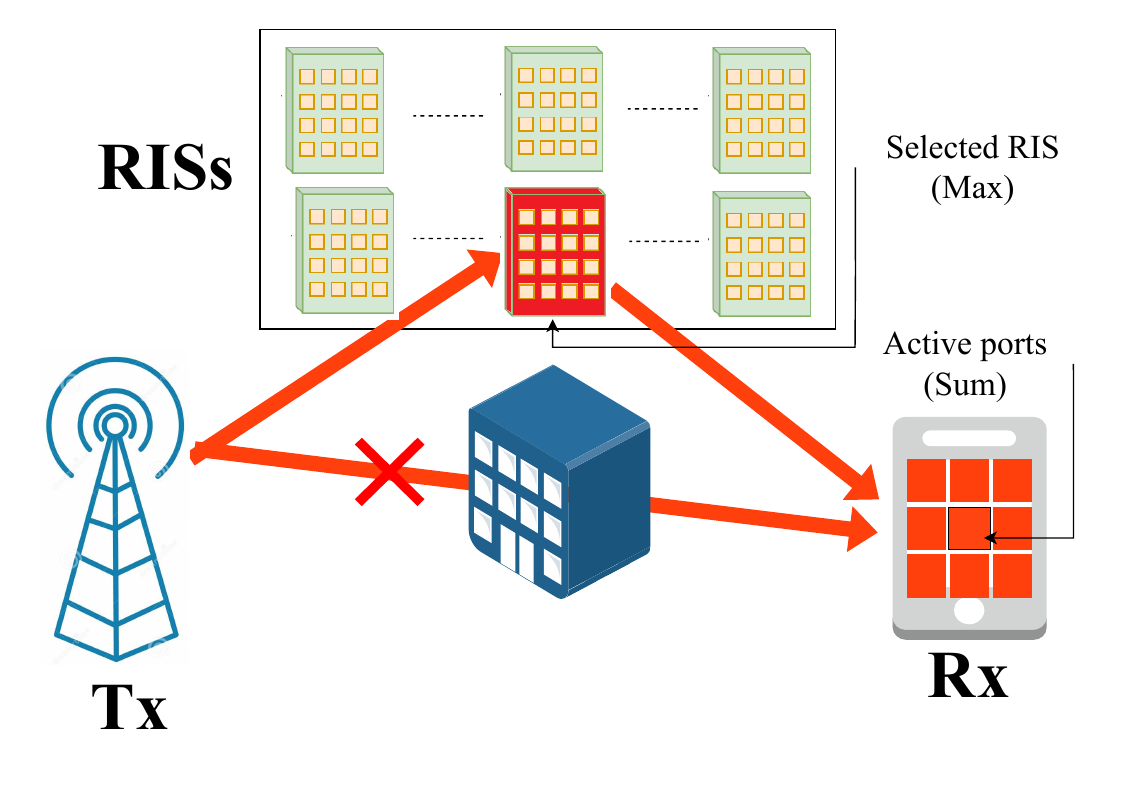}}
\caption{System model of multi-RIS-assisted FAS (a) \textit{Max-Max} selection scheme (b) \textit{Max-Sum} selection scheme.}
\label{nniyiyiy}
\end{figure}


\section{System Model}\label{System Model}

As illustrated in Fig. 1, the wireless communication system under consideration comprises a multi-RIS-assisted FAS, wherein a single-antenna transmitter $(\text{T}_{x})$ communicates with a planar FAS-equipped receiver $(\text{R}_{x})$, through the assistance of $N$ RISs. It is assumed that both the $\text{T}_{x}$ and the $\text{R}_{x}$ are located in the far-field region of the RISs. Furthermore, direct transmission links between the $\text{T}_{x}$ and the $\text{R}_{x}$ are obstructed due to the presence of natural or artificial barriers. Therefore, the communication between the $\text{T}_{x}$ and the $\text{R}_{x}$ occurs exclusively through RISs. 

In this paper, the $n$-th RIS is denoted as $\text{RIS}_{n}$, where $n\in\left\{1, \dots, N\right\}$, and consists of $M_{n}$ passive reflecting elements. The $i$-th element of the $\text{RIS}_{n}$ is denoted by $\text{RIS}_{n,i}$, where $i\in\left\{1, \dots, M_{n}\right\}$.
The $\text{R}_{x}$ is equipped with a two-dimensional (2D) surface, with $K$ pre-defined locations, referred to as ports, distributed across an area $A$. 
In addition, the $k$-th port of $\text{R}_{x}$ is denoted as $\text{R}_{x,k}$, where $k\in\left\{1, \dots, K\right\}$. 
In the proposed system model, a grid-based structure is adopted where $K_{i}$ port positions are uniformly distributed along a linear dimension of length $\lambda d_{i}$, with $i \in \{1, 2\}$, and $\lambda$ represents the wavelength of the carrier with frequency $f_c$. Hence, the total number of ports is $K = K_{1} \times K_{2}$, while the total area is $A = \lambda d_{1} \times \lambda d_{2}$. A suitable mapping function converts the 2D indices to the one-dimensional (1D) index, defined such that $f(u) = (u_{1}, u_{2})$, with its inverse $f^{-1}(u_{1}, u_{2}) = u$, where $u \in \{1, \ldots, K\}$, $u_{1} \in \{1, \ldots, K_{1}\}$, and $u_{2} \in \{1, \ldots, K_{2}\}$. 


The channel in the first hop (between the $\text{T}_{x}$ and $\text{RIS}_{n}$) undergoes Nakagami-$m$ fading. The channel vector between the $\text{T}_{x}$ and $\text{RIS}_{n}$ is denoted by $\mathbf{h}_{n}\in \mathbb{C}^{M_{n}\times1}$ and written as $\mathbf{h}_{n} = \left[h^{(1)}_{n}, \dots, h^{(i)}_{n}, \dots, h^{(M_{n})}_{n}\right]^{T}$, where $h^{(i)}_{n}=\frac{1}{L_{n,1}}\alpha^{(i)}_{n}e^{-j\phi^{(i)}_{n}}$ denotes the channel coefficient between the $\text{T}_{x}$ and RIS$_{n,i}$, where $L_{n,1}$ represents the distance from the $\text{T}_{x}$ to $\text{RIS}_{n}$, while $\alpha^{(i)}_{n}$ and $\phi^{(i)}_{n}$ refer to the channel amplitude and phase, respectively.
The channel in the second hop (between $\text{RIS}_{n}$ and the $\text{R}_{x}$) also undergoes Nakagami-$m$ fading. The channel vector between $\text{RIS}_{n}$ and $\text{R}_{x,k}$ is denoted by $\mathbf{g}_{n,k}\in \mathbb{C}^{M_{n}\times1}$ and written as $\mathbf{g}_{n,k} = \left[g^{(1)}_{n,k}, \dots, g^{(i)}_{n,k}, \dots, g^{(M_{n})}_{n,k}\right]^{T}$, where $g^{(i)}_{n,k} = \frac{1}{L_{n,2}}\beta^{(i)}_{n,k}e^{-j\Phi^{(i)}_{n}}$ denotes the channel coefficient between $\text{RIS}_{n,i}$ and $\text{R}_{x,k}$, where $L_{n,2}$ represents the distance from $\text{RIS}_{n}$ to $\text{R}_{x}$, while $\beta^{(i)}_{k,n}$ and $\Phi^{(i)}_{n}$ refer to the channel amplitude and phase, respectively. 
The reflection coefficients of the $\text{RIS}_{n}$ are represented as entries of the diagonal matrix $\mathbf{\Theta}_{n}\in \mathbb{C}^{M_{n}\times M_{n}}$. For the $i$-th element, under the full reflection assumption, we have $\mathbf{\Theta}^{(i,i)}_{n}=e^{j\theta^{(i)}_{n}}$, where $\theta^{(i)}_{n}\in[0,2\pi)$. 
Therefore, the received signal at $\text{R}_{x,k}$ from the reflected signals of $\text{RIS}_{n}$, denoted by $y_{k,n}$, can be written as
\begin{align} 
		\label{receivedsignal}
		\text{y}_{k,n} &= \mathbf{h}^{T}_{n}\mathbf{\Theta}_{n}\mathbf{g}_{n,k}{x}+{z}_{n} \nonumber\\
		& = \underbrace{\frac{1}{L_{n,1}L_{n,2}}\sum_{i = 1}^{M_{n}}\alpha^{(i)}_{n}\beta^{(i)}_{n,k}e^{j\left(\theta^{(i)}_{n}-\phi^{(i)}_{n}-\Phi^{(i)}_{n}\right)}{x}}_{\text{signal}} +\underbrace{{z}_{n}}_{\text{noise}}, 
\end{align}
where ${x}$ is the transmitted signal with average power $P_{x}$, i.e., $\mathbb{E}\left[|x|^2\right] = P_{x}$, where $\mathbb{E}\left[\cdot\right]$ and $|\cdot|$ refer to the expected value and magnitude, respectively, and ${z}_{n}$ denotes the additive white Gaussian noise (AWGN) sample with zero mean and variance $\sigma^{2}$.
	
The SNR is defined as the ratio of the power of the signal to the power of the noise, i.e., $\text{SNR} = \frac{P_{\text{signal}}}{P_{\text{noise}}}$. Thus, from \eqref{receivedsignal}, the received SNR corresponding to the $k$-th port and the $\text{RIS}_{n}$, denoted by $\gamma_{k, n}$, can be expressed  as 
\begin{align} \label{eq:gamma1}
		\gamma_{k, n} &= \frac{\mathbb{E}\left[\bigg|\frac{1}{L_{n,1}L_{n,2}}\sum_{i = 1}^{M_{n}}\alpha^{(i)}_{n}\beta^{(i)}_{n,k}e^{j\left(\theta^{(i)}_{n}-\phi^{(i)}_{n}-\Phi^{(i)}_{n}\right)}{x}\bigg|^2\right]}{\mathbb{E}\left[|{z}_{n}|^2\right]}\nonumber\\
		& = \frac{1}{L^{2}_{n,1}L^{2}_{n,2}}\frac{\mathbb{E}\left[|{x}|^2\right]}{\mathbb{E}\left[|{z}_{n}|^2\right]}\bigg|\sum_{i = 1}^{M_{n}}\alpha^{(i)}_{n}\beta^{(i)}_{n,k}e^{j\left(\theta^{(i)}_{n}-\phi^{(i)}_{n}-\Phi^{(i)}_{n}\right)}\bigg|^2.
\end{align}
Define $\bar{\gamma}$ as the average SNR, i.e., $\bar{\gamma} = \frac{\mathbb{E}\left[|{x}|^2\right]}{\mathbb{E}\left[|{z}_{n}|^2\right]} = \frac{P_{x}}{\sigma^{2}}$. Then, $\gamma_{k, n}$ in \eqref{eq:gamma1} can be written as
	\begin{align} \label{eq:gamma2}
		\gamma_{k, n} &= \frac{\bar{\gamma}}{L^{2}_{n,1}L^{2}_{n,2}}\bigg|\sum_{i = 1}^{M_{n}}\alpha^{(i)}_{n}\beta^{(i)}_{n,k}e^{j\left(\theta^{(i)}_{n}-\phi^{(i)}_{n}-\Phi^{(i)}_{n}\right)}\bigg|^2\nonumber\\
		& \overset{(a)}{=}\frac{\bar{\gamma}}{L^{2}_{n,1}L^{2}_{n,2}}\left ({\sum _{i=1}^{M_{n}} \alpha^{(i)}_{n}\beta^{(i)}_{n,k}}\right)^{2}.
	\end{align}
	Here $(a)$ is obtained from the assumption of perfect CSI for the RIS configuration. Thus, the $\text{RIS}_{n}$ aligns the phases of the reflected signals to the sum of the phases of its incoming and outgoing fading channels, i.e., $\theta^{(i)}_{n} = \phi^{(i)}_{n} + \Phi^{(i)}_{n}$.
\section{Selection Schemes}\label{Selection Scheme}
In this section, two selection schemes are introduced, referred to as \textit{Max-Max} and \textit{Max-Sum}. The algorithms of these schemes are presented in \textbf{Algorithms \ref{Max-Max selection scheme}} and \textbf{\ref{Max-Sum selection scheme}}, respectively. The formal definitions and corresponding mathematical formulations are as follows:
\subsection{\textit{Max-Max} Scheme} In this proposed selection scheme, the objective is to select the best RIS and port combination that maximizes the SNR at the $\text{R}_x$. Specifically, one RIS from a set of $N$ available RISs is chosen, and one port from a set of $K$ available ports are selected to achieve the maximum received SNR at the $\text{R}_x$. 
To accomplish this, the process involves two steps. In the first step, for each port, the best RIS that offers the highest SNR is selected, i.e.,  
		\begin{equation}\label{max-max_eq1}
			n^{*} = \underset{n = 1, \dots, N}{\arg\max} \left\{\gamma_{k,n}\right\},
		\end{equation}
		where $n^*$ denotes the index of the best RIS. Consequently, the SNR corresponding to the $k$-th port and the $\text{RIS}_{n^*}$ is expressed as
		\begin{equation}\label{max-max_eq2}
			\gamma_{k,n^{*}} = \underset{n = 1, \dots, N}{\max} \left\{\gamma_{k,n}\right\}. 
		\end{equation}

\begin{algorithm} [t]
	\caption{Proposed \textit{Max-Max} selection scheme.}\label{Max-Max selection scheme}
	\begin{algorithmic}[1]
		\State \textbf{Input:} A set of $N$ RISs, a set of $K$ ports, and the initial SNR matrix, $\gamma_{k,n}$, for $k = 1, \dots, K$ and $n = 1, \dots, N$.
		\State \textbf{Output:} The index of the best port $k^*$, the index of the best RIS $n^*$, and the corresponding maximum achievable SNR, $\gamma_{k^*,n^*}$.
		\For{$k = 1$ to $K$}
		\For{$n = 1$ to $N$}
		\State Evaluate $\gamma_{k,n}$ as per the relation defined in \eqref{eq:gamma2}.
		\EndFor
		\EndFor
		
		\State Identify the RIS index $n^*$ that maximizes the SNR for each $k$-th port, i.e.,
		\[
		n^* \gets \underset{n = 1, \dots, N}{\argmax} \left(\gamma_{k,n}\right) \quad \text{(from \eqref{max-max_eq1})}.
		\]
		
		\State Calculate the maximum SNR for the selected $n^*$-th RIS at the $k$-th port, i.e.,
		\[
		\gamma_{k,n^*} \gets \underset{n = 1, \dots, N}{\max} \left\{\gamma_{k,n}\right\} \quad \text{(from \eqref{max-max_eq2})}.
		\]
		
		\State Determine the port $k^*$ that provides the maximum SNR for the selected $n^*$-th RIS, i.e.,
		\[
		k^* \gets \underset{k = 1, \dots, K}{\argmax} \left\{\gamma_{k,n^{*}}\right\} \quad \text{(from \eqref{max-max_eq1+})}.
		\]
		
		\State Finally, compute the maximum achievable SNR at the receiver by selecting the best port and RIS combination, i.e.,
		\[
		\gamma_{k^*,n^*} \gets \underset{k = 1, \dots, K}{\max} \left\{\gamma_{k,n^*}\right\} \quad \text{(from \eqref{max-max_eq2+})}.
		\]
		
		\State \textbf{Return:} The best port index $k^*$, the best RIS index $n^*$, and the corresponding maximum achievable SNR, $\gamma_{k^*,n^*}$.
	\end{algorithmic}
\end{algorithm}

		In the second step, the best port that provides the maximum SNR from the set of $\gamma_{k,n^{*}}$ values is selected (activated), i.e.,
		\begin{equation}\label{max-max_eq1+}
			k^{*} = \underset{k = 1, \dots, K}{\arg\max} \left\{\gamma_{k,n^{*}}\right\},
		\end{equation}
		where $k^*$ denotes the index of the best port. 
		
		Finally, the maximum achievable SNR at the $\text{R}_x$ for the selected indices $(k^*, n^*)$ is given by
		\begin{equation}\label{max-max_eq2+}
			\gamma_{\textit{Max-Max}} = \underset{k = 1, \dots, K}{\max} \left\{\gamma_{k,n^{*}}\right\}.
		\end{equation}
			
		\subsection{\textit{Max-Sum} Scheme} In this selection scheme, the system selects one RIS from the $N$ available RISs, while all $K$ ports are activated to combine the signals received from the chosen RIS at the receiver, $\text{R}_x$. To achieve this, the first step is to combine the received signals from each $\text{RIS}_{n}$ across all ports, i.e.,  
		\begin{equation}\label{max-sum_eq2_+}
			S_{n} = \sum_{k = 1}^{K} \gamma_{k,n}.
		\end{equation}
		Next, the best RIS is selected by identifying the one that offers the maximum summation of SNR across all ports, i.e.,
		\begin{equation}\label{max-sum_eq1}
			n^{*} = \underset{\substack{n = 1, \dots, N}}{\arg\max} \left\{S_{n}\right\}.
		\end{equation}
		Thus, the maximum sum of the received SNR at the receiver $\text{R}_x$ is represented as
		\begin{equation}\label{max-sum_eq2}
			\gamma_{\textit{Max-Sum}} = \underset{\substack{n = 1, \dots, N}}{\max} \left\{S_{n}\right\}.
		\end{equation}

\begin{algorithm} [t]
	\caption{Proposed \textit{Max-Sum} selection scheme.}\label{Max-Sum selection scheme}
	\begin{algorithmic}[1]
		\State \textbf{Input:} A set of $N$ RISs, a set of $K$ ports, and the initial SNR matrix, $\gamma_{k,n}$, for $k = 1, \dots, K$ and $n = 1, \dots, N$.
		\State \textbf{Output:} The index of the best RIS $n^*$, and the maximum sum of the received SNR at the receiver $\text{R}_x$, $\gamma_{\textit{Max-Sum}}$.
		
		\For{$n = 1$ to $N$}
		\State Combine the SNR values from all ports for each $\text{RIS}_n$, i.e.,
		\[
		S_{n} = \sum_{k = 1}^{K} \gamma_{k,n} \quad \text{(from \eqref{max-sum_eq2_+})}.
		\]
		\EndFor
		
		\State Select the best RIS $n^*$ that maximizes the total SNR across all ports, i.e.,
		\[
		n^{*} = \underset{\substack{n = 1, \dots, N}}{\arg\max} \left\{ S_{n} \right\} \quad \text{(from \eqref{max-sum_eq1})}.
		\]
		
		\State Compute the maximum summation of SNR for the selected RIS $n^*$, i.e.,
		\[
		\gamma_{\textit{Max-Sum}} = \underset{\substack{n = 1, \dots, N}}{\max} \left\{ S_{n} \right\} \quad \text{(from \eqref{max-sum_eq2})}.
		\]
		
		\State \textbf{Return:} The best RIS index $n^*$, and the corresponding maximum sum of SNR at the receiver, $\gamma_{\textit{Max-Sum}}$.
	\end{algorithmic}
\end{algorithm}


\section{Performance Analysis} \label{Performance Analysis}
In this section, we present the performance analysis of multi-RIS-assisted FAS, incorporating two novel selection schemes. To this end, we first obtain the CDF of the general expression for the SNR. Subsequently, the CDFs of the SNR for both proposed schemes are obtained. Thereafter, we introduce the system's OP and DOR. Finally, to gain deeper insights into the system's performance, we conducted an asymptotic analysis in the high-SNR regime, highlighting the diversity and coding gains achieved by the system.

\subsection{CDF of SNR}
In this study, we postulate that the variables $\alpha^{(i)}_{n}$ and $\beta^{(i)}_{n,k}$, as defined in \eqref{eq:gamma2}, follow the Nakagami-$m$ distributions. These random variables are characterized by the shape parameter, $m$, and the scale parameter, $\Omega$, which collectively determine their statistical behavior. According to \cite{CDF_theory}, the exact distribution of $\gamma_{k,n}$ in \eqref{eq:gamma2} can be effectively approximated by the first term of a Laguerre series expansion, yielding a Gamma distribution. Consequently, the CDF of $\gamma_{k,n}$ can be expressed as
		\begin{align}\label{eq:CDFgamma}
			F_{\gamma_{k,n}}(\gamma)= \frac{\Upsilon \left(a_{k,n},\sqrt{\frac{L^{2}_{n,1}L^{2}_{n,2} }{b^{2}_{k,n}}\frac{\gamma}{\bar{\gamma}}}\right)}{\Gamma (a_{k,n})}.
		\end{align}
        %
Here, $\Upsilon(\cdot, \cdot)$ denotes the lower incomplete Gamma function, defined as $\Upsilon(s, x) = \int_0^x t^{s-1} e^{-t} dt$, while $\Gamma(\cdot)$ represents the Gamma function, given by $\Gamma(s) = \int_0^\infty t^{s-1} e^{-t}dt$. Moreover, $m_{k,n,i}$ and $\Omega_{k,n,i}$ correspond to the shape and scale parameters of the Nakagami-\emph{m} distribution for the first and second transmission hops, where $i \in \{1, 2\}$. Furthermore, the parameters $a_{k,n}$ and $b_{k,n}$ are defined in \eqref{eq:a} and \eqref{eq:b}, respectively. 
		\begin{figure*}[b!]
			\begin{align}\label{eq:a}
				a_{k,n} = \frac{m_{n,1} m_{k,n,2} M_{n} \Gamma (m_{n,1})^2 \Gamma (m_{k,n,2})^2}{m_{n,1} m_{k,n,2} \Gamma (m_{n,1})^2 \Gamma (m_{k,n,2})^2 - \Gamma \left(m_{n,1} + \frac{1}{2}\right)^2 \Gamma \left(m_{k,n,2} + \frac{1}{2}\right)^2} - M_{n}.
			\end{align}
			\hrule
			\begin{align}\label{eq:b}
				b_{k,n} = \frac{m_{n,1} m_{k,n,2} \Gamma (m_{n,1})^2 \Gamma (m_{k,n,2})^2 - \Gamma \left(m_{n,1} + \frac{1}{2}\right)^2 \Gamma \left(m_{k,n,2} + \frac{1}{2}\right)^2}{\sqrt{\frac{m_{n,1}}{\Omega_{n,1}}} \Gamma (m_{n,1}) \Gamma \left(m_{n,1} + \frac{1}{2}\right) \sqrt{\frac{m_{k,n,2}}{\Omega_{k,n,2}}} \Gamma (m_{k,n,2}) \Gamma \left(m_{k,n,2} + \frac{1}{2}\right)}.
			\end{align}
			\hrule
		\end{figure*}

\subsection{CDF of the SNR of the Max-Max selection scheme}
To derive the CDF of the \textit{Max-Max} selection scheme, it is first necessary to obtain the CDF of $\gamma_{k,n^{*}}$. Subsequently, the CDF of $\gamma_{\textit{Max-Max}}$ can be determined. In this regard, propositions \ref{proposition_CDF++} and \ref{proposition_CDF_Max-Max} present the CDF expressions for $\gamma_{k,n^{*}}$ and $\gamma_{\textit{Max-Max}}$, respectively. 
\begin{proposition}\label{proposition_CDF++}
	The CDF of $\gamma_{k,n^{*}}$, as defined in \eqref{max-max_eq2}, denoted by  $F_{\gamma_{k,n^{*}}}\left(\gamma\right)$, is formulated as
		\begin{align} \label{eq:CDF5+}
			F_{\gamma_{k,n^{*}}}(\gamma) &=\prod_{n = 1}^{N}\sum_{j = 0}^{\infty} \frac{\left(-1\right)^{j}\left(\frac{L^{2}_{n,1}L^{2}_{n,2} }{b^{2}_{k,n}}\frac{\gamma}{\bar{\gamma}}\right)^{\frac{a_{k,n}+j}{2}}}{j!\left(a_{k,n} + j\right)\Gamma (a_{k,n})}.
		\end{align}
	\end{proposition}
\begin{proof}
	To determine the $F_{\gamma_{k,n^{*}}}\left(\gamma\right)$, we need to compute the cumulative probability that $\gamma_{k,n^{*}}$ is less than or equal to a given value $\gamma$. From \eqref{max-max_eq2}, $\gamma_{k,n^{*}}$ can be considered as the maximum of a set of independent random variables, $\gamma_{k,n}$. Thus, $F_{\gamma_{k,n^{*}}}\left(x\right)$ can be expressed as a product of the CDFs of the individual $\gamma_{k,n}$ values, i.e.,
	\begin{align} \label{eq:CDF1+}
		F_{\gamma_{k,n^{*}}}(\gamma) &= \mathbb{P}\left(\gamma_{k,n^{*}} \leq \gamma\right) \nonumber\\
		&= \mathbb{P}\left(\underset{n = 1, \dots, N}{\max}\left\{\gamma_{k,n}\right\} \leq \gamma\right)\nonumber\\
		&= \mathbb{P}\left(\max\left(\gamma_{k,1}, \dots, \gamma_{k,N}\right) \leq \gamma\right)\nonumber\\
		&= \mathbb{P}\left(\gamma_{k,1}\leq \gamma, \dots, \gamma_{k,N} \leq \gamma\right)\nonumber\\
		&= \mathbb{P}\left(\gamma_{k,1}\leq \gamma\right) \dots \mathbb{P}\left(\gamma_{k,N} \leq \gamma\right)
		\nonumber\\
		&=F_{\gamma_{k,1}}(\gamma) \dots F_{\gamma_{k,N}}(\gamma)
		\nonumber\\
		&=\prod_{n = 1}^{N}F_{\gamma_{k,n}}(\gamma).
	\end{align}
Here $\mathbb{P}\left(\cdot\right)$ refers to the probability symbol. 
%
By substituting \eqref{eq:CDFgamma} into \eqref{eq:CDF1+}, we obtain
	\begin{align} \label{eq:CDF2+}
		F_{\gamma_{k,n^{*}}}(\gamma) &=\prod_{n = 1}^{N}\frac{\Upsilon \left(a_{k,n},\sqrt{\frac{L^{2}_{n,1}L^{2}_{n,2} }{b^{2}_{k,n}}\frac{\gamma}{\bar{\gamma}}}\right)}{\Gamma (a_{k,n})}.
	\end{align}
	However, the lower incomplete Gamma function in \eqref{eq:CDF2+} can be expanded into series \cite{Grad}, i.e.,
	\begin{align} \label{eq:CDF3+}
		\Upsilon\left(s, x \right) = \sum_{j = 0}^{\infty} \frac{\left(-1\right)^{j}x^{s+j}}{j!\left(s + j\right)}.
	\end{align}
	By substituting \eqref{eq:CDF3+} into \eqref{eq:CDF2+}, we then obtain
	\begin{align} \label{eq:CDF4+}
		F_{\gamma_{k,n^{*}}}(\gamma) &=\prod_{n = 1}^{N}\sum_{j = 0}^{\infty} \frac{\left(-1\right)^{j}\left(\sqrt{\frac{L^{2}_{n,1}L^{2}_{n,2} }{b^{2}_{k,n}}\frac{\gamma}{\bar{\gamma}}}\right)^{a_{k,n}+j}}{j!\left(a_{k,n} + j\right)\Gamma (a_{k,n})}.
	\end{align} 
	With some mathematical simplifications, we obtain Proposition \ref{proposition_CDF++}, and the proof ends.
\end{proof}

\begin{corollary}\label{corollary1}
If we assume that the $N$ random variables in Proposition \ref*{proposition_CDF++} are identical, the parameters $L_{n,1}$, $L_{n,2}$, $b_{k,n}$, and $a_{k,n}$ will be the same for all $n$, i.e., $L_{n,1} = L_{1}$, $L_{n,2} = L_{2}$, $a_{k,n} = a_k$, and $b_{k,n} = b_k$. In this case, the expression simplifies as follows
\begin{align} \label{eq:CDF5++}
	F_{\gamma_{k,n^{*}}}(\gamma) &= \left(\sum_{j = 0}^{\infty} \frac{\left(-1\right)^{j}}{j! \left(a_k + j\right) \Gamma(a_k)} \left(\frac{L^{2}_{1}L^{2}_{2}}{b_k^2} \frac{\gamma}{\bar{\gamma}}\right)^{\frac{a_k + j}{2}} \right)^N.
\end{align}
\end{corollary}
\begin{corollary}\label{corollary2}
Following Corollary \ref{corollary1}, if we assume that only one port is available on the FAS-assisted $\text{R}_{x}$, i.e., $k = 1$, the parameters are assumed to be $a_{k} = a$ and $b_{k} = b$. In this case, a multi-RIS-assisted system is considered, and the expression in Corollary \ref{corollary1} simplifies to an exact closed-form as follows
	\begin{align} \label{eq:CDF5+++}
		F_{\gamma_{\text{Max}}}(\gamma) &= \left( \sum_{j = 0}^{\infty} \frac{\left(-1\right)^{j}}{j! \left(a + j\right) \Gamma(a)} \left(\frac{L^{2}_{1}L^{2}_{2}}{b^2} \frac{\gamma}{\bar{\gamma}}\right)^{\frac{a + j}{2}} \right)^N.
	\end{align}
\end{corollary}
	\begin{proposition}\label{proposition_CDF_Max-Max}
		The CDF of $\gamma_{\textit{Max-Max}}$, as defined in \eqref{max-max_eq2+}, denoted by $F_{\gamma_{\textit{Max-Max}}}\left(\gamma\right)$ can be expressed as 
			\begin{align} \label{eq:FinalExpression}
				&F_{\gamma_{\textit{Max-Max}}}(\gamma)\\ &= \Phi_{\Sigma}\left( 
				\Phi^{-1}\left(
				\prod_{n=1}^{N} \sum_{j=0}^{\infty} 
				\frac{\left(-1\right)^j \left( \frac{L^{2}_{n,1}L^{2}_{n,2}}{b^2_{1,n}} \frac{\gamma}{\bar{\gamma}} \right)^{\frac{a_{1,n} + j}{2}}}{j! \, \left(a_{1,n} + j\right) \Gamma(a_{1,n})}
				\right), \right. \nonumber \\
				&\quad \left. \ldots, \right. \nonumber
				\quad \left. \Phi^{-1}\left(
				\prod_{n=1}^{N} \sum_{j=0}^{\infty} 
				\frac{\left(-1\right)^j \left( \frac{L^{2}_{n,1}L^{2}_{n,2}}{b^2_{K,n}} \frac{\gamma}{\bar{\gamma}} \right)^{\frac{a_{K,n} + j}{2}}}{j! \, \left(a_{K,n} + j\right) \Gamma(a_{K,n})}
				\right)
				\right). 
			\end{align}
		Here $\Phi^{-1}$ is the inverse of the standard normal CDF. Additionally, $\Phi_{\Sigma}$ is the multivariate normal CDF with mean vector $\bm{0}$ and correlation matrix $\Sigma$, which captures the dependencies between the variables, $\gamma_{k,n^*}$, and is given as 
		\begin{equation}
			\Sigma =
			\begin{bmatrix}
				\eta_{1,1} & \eta_{1,2} & \dots & \eta_{1,K} \\
				\eta_{2,1} & \eta_{2,2} & \dots & \eta_{2,K} \\
				\vdots & \vdots & \ddots & \vdots \\
				\eta_{K,1} & \eta_{K,2} & \dots & \eta_{K,K}
			\end{bmatrix},
		\end{equation}
		where $\eta_{u, \tilde{u}}$ is the spatial correlation between any two ports, $u = f^{-1}(u_1, u_2)$ and $\tilde{u} = f^{-1}(\tilde{u}_1, \tilde{u}_2)$, and can be characterized as
		\begin{align}
			\eta_{u,\tilde{u}} = \text{sinc}\left(
			2 \frac{|u_1 - \tilde{u}_1|}{K_1 - 1} d_1 +
			2 \frac{|u_2 - \tilde{u}_2|}{K_2 - 1} d_2
			\right),
		\end{align}
		in which $\text{sinc}(t) = \frac{\sin(\pi t)}{\pi t}$ is the sinc function. 
	\end{proposition}

	\begin{proof}
	To determine $F_{\gamma_{\textit{Max-Max}}}\left(\gamma\right)$, we need to compute the cumulative probability that $\gamma_{\textit{Max-Max}}$ is less than or equal to a given value $\gamma$
	\begin{align} \label{eq:CDF1}
		F_{\gamma_{\textit{Max-Max}}}(\gamma) &= \mathbb{P}\left(\gamma_{\textit{Max-Max}} \leq \gamma\right) \nonumber\\
		&= \mathbb{P}\left(\underset{k = 1, \dots, K}{\max} \left\{\gamma_{k,n^{*}}\right\} \leq \gamma\right)\nonumber\\
		&= \mathbb{P}\left(\max\left(\gamma_{1,n^{*}}, \dots, \gamma_{K,n^{*}}\right) \leq \gamma\right)\nonumber\\
		&= F_{\gamma_{1,n^{*}}, \dots, \gamma_{K,n^{*}}}\left(\gamma\right),
	\end{align}
where $F_{\gamma_{1,n^*}, \gamma_{2,n^*}, \ldots, \gamma_{K,n^*}}$ is the joint CDF of $\gamma_{1,n^*}, \ldots, \gamma_{K,n^*}$. The joint CDF of the $K$-dimensional random vector, $\bm{\gamma} = (\gamma_{1,n^*}, \gamma_{2,n^*}, \ldots, \gamma_{K,n^*})$, can be expressed in terms of a copula function $\mathbb{C}$ as
\begin{align}\label{eq:CDF*}
	F_{\gamma_{1,n^*}, \gamma_{2,n^*}, \ldots, \gamma_{K,n^*}}(\gamma) &= \mathbb{C}\left(F_{\gamma_{1,n^*}}(\gamma),\right. \nonumber \\
	&\quad \left. \ldots, F_{\gamma_{K,n^*}}(\gamma)\right).
\end{align}
For the Gaussian copula, the copula function is defined as
\begin{align}\label{eq:CDF**}
	\mathbb{C}(u_1, u_2, \ldots, u_K) &= \Phi_{\Sigma}\left(\Phi^{-1}(u_1), \Phi^{-1}(u_2),\right. \nonumber \\
	&\quad \left. \ldots, \Phi^{-1}(u_K)\right),
\end{align}
where $u_k = F_{\gamma_{k,n^*}}(\gamma)$ are the marginal CDF values for $k = 1, \ldots, K$. By substituting the Gaussian copula in \eqref{eq:CDF**} into the joint CDF in \eqref{eq:CDF*}, we obtain
\begin{align}\label{key}
	F_{\gamma_{\textit{Max-Max}}}(\gamma) &= \mathbb{C}\left(F_{\gamma_{1,n^*}}(\gamma), F_{\gamma_{2,n^*}}(\gamma), \ldots, F_{\gamma_{K,n^*}}(\gamma)\right) \nonumber \\
	&= \Phi_{\Sigma}\left(\Phi^{-1}\left(F_{\gamma_{1,n^*}}(\gamma)\right), \Phi^{-1}\left(F_{\gamma_{2,n^*}}(\gamma)\right), \right. \nonumber \\
	&\quad \left.\ldots, \Phi^{-1}\left(F_{\gamma_{K,n^*}}(\gamma)\right)\right).
\end{align}
Utilizing Proposition \ref{proposition_CDF++} in \eqref{key}, we obtain the final expression for $F_{\gamma_{\textit{Max-Max}}}(\gamma)$ in Proposition \ref{proposition_CDF_Max-Max}, and the proof ends. 
\end{proof}

\begin{corollary}
Assuming that all ports are independently distributed, then, from \eqref{eq:CDF1}, 
\begin{align} \label{eq:CDF1iid}
	F_{\gamma_{\text{Max-Max}}}(\gamma) &= \mathbb{P}\left(\max\left(\gamma_{1,n^{*}}, \dots, \gamma_{K,n^{*}}\right) \leq \gamma\right)\nonumber\\
	&=\mathbb{P}\left(\gamma_{1,n^{*}}\leq \gamma, \dots, \gamma_{K,n^{*}}\leq \gamma \right)\nonumber\\
	&=\mathbb{P}\left(\gamma_{1,n^{*}}\leq \gamma\right) \dots \mathbb{P}\left(\gamma_{K,n^{*}}\leq \gamma \right)\nonumber\\
	&=F_{\gamma_{1,n^{*}}}\left(\gamma\right)\dots F_{\gamma_{K,n^{*}}}\left(\gamma\right)\nonumber\\
	&=\prod_{k =1}^{K}F_{\gamma_{k,n^{*}}}\left(\gamma\right).
\end{align}
Using \eqref{eq:CDF2+}
	\begin{align} \label{eq:CDF2+iid}
	F_{\gamma_{\text{Max-Max}}}(\gamma) &=\prod_{k =1}^{K}\prod_{n = 1}^{N}\frac{\Upsilon \left(a_{k,n},\sqrt{\frac{L^{2}_{n,1}L^{2}_{n,2} }{b^{2}_{k,n}}\frac{\gamma}{\bar{\gamma}}}\right)}{\Gamma (a_{k,n})}.
\end{align}
In the case of the identical parameters, then 
	\begin{align} \label{eq:CDF2+iidd}
	F_{\gamma_{\text{Max-Max}}}(\gamma) &=\left(\frac{\Upsilon \left(a,\sqrt{\frac{L^{2}_{1}L^{2}_{2}}{b^{2}}\frac{\gamma}{\bar{\gamma}}}\right)}{\Gamma (a)}\right)^{KN}\nonumber\\
	&=\left(\sum_{j = 0}^{\infty} \frac{\left(-1\right)^{j}}{j! \left(a + j\right) \Gamma(a)} \left(\frac{L^{2}_{1}L^{2}_{2}}{b^2} \frac{\gamma}{\bar{\gamma}}\right)^{\frac{a + j}{2}}\right)^{KN}.
\end{align}
\end{corollary}

\subsection{CDF of \textit{Max-Sum} selection scheme}
To derive the CDF of the \textit{Max-Sum} selection scheme, we begin by obtaining the CDF of $S_{n}$. Once this is determined, we can proceed to find the CDF of $\gamma_{\textit{Max-Sum}}$. In this context, Propositions \ref{proposition_CDF_S_n} and \ref{proposition_CDF_Max-Sum} provide the CDF expressions for $S_{n}$ and $\gamma_{\textit{Max-Sum}}$, respectively.

\begin{proposition}\label{proposition_CDF_S_n}
The CDF of $S_n$, as defined in \eqref{max-sum_eq2_+} and denoted by $F_{S_{n}}\left(\gamma\right)$, can be expressed as 
\begin{equation}\label{S_n}
    F_{S_n}(\gamma) = \sum_{j = 0}^{\infty} \frac{\left(-1\right)^{j}\left(\sum_{k=1}^{K}\frac{L^{2}_{n,1}L^{2}_{n,2} }{b^{2}_{k,n}}\frac{\gamma}{\bar{\gamma}}\right)^{\frac{\sum_{k=1}^{K} a_{k,n}+j}{2}}}{j!\left(\sum_{k=1}^{K} a_{k,n} + j\right)\Gamma (\sum_{k=1}^{K} a_{k,n})}.
\end{equation}
\end{proposition}
\begin{proof}
    Finding the exact expression for $F_{S_{n}}\left(\gamma\right)$ is intractable. Here, we derive an approximate expression by assuming independent Gamma-distributed RVs. Therefore, $S_n$ in \eqref{max-sum_eq2_+} is the sum of independent Gamma-distributed RVs, and its CDF can be calculated as
\begin{equation}
    F_{S_n}(\gamma) = \mathbb{P}(S_n \le \gamma).
\end{equation}
We now use the property that the sum of independent Gamma-distributed RVs follows another Gamma distribution, i.e, 
\begin{equation}\label{S_n_1}
    F_{S_n}(\gamma) = \frac{\Upsilon \left(\sum_{k=1}^{K} a_{k,n}, \sqrt{\sum_{k=1}^{K}\frac{L^{2}_{n,1}L^{2}_{n,2}}{b_{k,n}^{2}} \frac{\gamma}{\bar{\gamma}}} \right)}{\Gamma \left( \sum_{k=1}^{K} a_{k,n} \right)}.
\end{equation}
By expanding the lower incomplete Gamma function in \eqref{S_n_1} into a series as in \eqref{eq:CDF3+}, we obtain the corresponding CDF in Proposition \ref{proposition_CDF_S_n}, and the proof ends.
\end{proof}

\begin{proposition}\label{proposition_CDF_Max-Sum}
The CDF of $\gamma_{\textit{Max-Sum}}$, as defined in \eqref{max-sum_eq2} and denoted by $F_{\gamma_{\textit{Max-Sum}}}\left(\gamma\right)$, is given by
\begin{equation}
    F_{\gamma_{\textit{Max-Sum}}}(\gamma) = \prod_{n=1}^{N} \sum_{j = 0}^{\infty} \frac{\left(-1\right)^{j}\left(\sum_{k=1}^{K}\frac{L^{2}_{n,1}L^{2}_{n,2} }{b^{2}_{k,n}}\frac{\gamma}{\bar{\gamma}}\right)^{\frac{\sum_{k=1}^{K} a_{k,n}+j}{2}}}{j!\left(\sum_{k=1}^{K} a_{k,n} + j\right)\Gamma (\sum_{k=1}^{K} a_{k,n})}.
\end{equation}
\end{proposition}
\begin{proof}
From \eqref{max-sum_eq2}, $\gamma_{\textit{Max-Sum}}$ is the maximum over $N$ independent RVs $S_n$. Hence, by using the order statistics result, the CDF $F_{\gamma_{\textit{Max-Sum}}}$ can be calculated as
\begin{equation}\label{444}
    F_{\gamma_{\textit{Max-Sum}}}(\gamma) = \mathbb{P} \left( \max_{n=1,\dots,N} S_n \le \gamma \right).
\end{equation}
Since the RVs $S_n$ are independent, 
\begin{equation}\label{rrr}
    F_{\gamma_{\textit{Max-Sum}}}(\gamma) = \prod_{n=1}^{N} F_{S_n}(\gamma).
\end{equation}
Substituting \eqref{S_n} into \eqref{rrr}, we obtain Proposition \ref{proposition_CDF_Max-Sum}, and the proof ends.
\end{proof}

\begin{corollary}\label{corollary4}
For identical RVs (i.e., $ L_{n,1} = L_{1}$, $L_{n,2} = L_{2}$, $a_{k,n} = a$, and $b_{k,n} = b$), the expression in Proposition \ref{proposition_CDF_Max-Sum} simplifies to
\begin{equation}\label{max-sum-iid}
    F_{\gamma_{\textit{Max-Sum}}}(\gamma) = \left(\sum_{j = 0}^{\infty} \frac{\left(-1\right)^{j} \left(K\frac{L^{2}_{1}L^{2}_{2}}{b^2} \frac{\gamma}{\bar{\gamma}}\right)^{\frac{Ka + j}{2}} }{j! \left(Ka + j\right) \Gamma(Ka)}\right)^N.
\end{equation}
\end{corollary}
\subsection{OP Analysis}
The OP is a key performance metric in wireless 
systems. It represents the probability that the instantaneous SNR falls below a predefined threshold, denoted as $\gamma_{\text{th}}$. This event occurs when the communication link cannot meet the desired quality of service. Mathematically, the OP is defined as
\begin{equation}\label{OP1}
	P_{\text{out}} = \mathbb{P}(\gamma < \gamma_{\text{th}}).
\end{equation}
In Proposition \ref{eq:OP}, we provide the OP expressions for both selection schemes.
\begin{proposition}\label{eq:OP}
	The OP for the \textit{Max-Max} and \textit{Max-Sum} selection schemes are given, respectively by
	\begin{align} \label{eq:OP_max_max}
		P^{\textit{Max-Max}}_{out}&=F_{\gamma_{\textit{Max-Max}}}(\gamma_{\text{th}}),
	\end{align}
and
	\begin{align} \label{eq:OP_max_sum}
		P^{\textit{Max-Sum}}_{out} &=F_{\gamma_{\textit{Max-Sum}}}(\gamma_{\text{th}}).
	\end{align}
\end{proposition}
\begin{proof}
From \eqref{OP1}, $P^{\textit{Max-Max}}_{out}$ and $P^{\textit{Max-Sum}}_{out}$ can be calculated, respectively as
\begin{align} \label{eq:OP_max_max1}
	P^{\textit{Max-Max}}_{out}&=\mathbb{P}(\gamma_{\textit{Max-Max}} < \gamma_{\text{th}})\nonumber\\
	&=	F_{\gamma_{\textit{Max-Max}}}(\gamma_{\text{th}}),
\end{align}
and
	\begin{align} \label{eq:OP_max_sum1}
	P^{\textit{Max-Sum}}_{out} &=\mathbb{P}(\gamma_{\textit{Max-Sum}} < \gamma_{\text{th}})\nonumber\\
	&=	F_{\gamma_{\textit{Max-Sum}}}(\gamma_{\text{th}}).
\end{align}
By substituting Propositions \ref{proposition_CDF_Max-Max} and \ref{proposition_CDF_Max-Sum} into \eqref{eq:OP_max_max1} and \eqref{eq:OP_max_sum1}, respectively, we obtain Proposition \ref{eq:OP}, and the proof ends.
	\end{proof}

\subsection{Asymptotic OP Analysis}
In wireless communication systems, it is crucial to analyze the OP's behavior as the SNR approaches infinity. Accordingly, in this subsection, we conduct the asymptotic OP analysis. This analysis provides key insights into system performance under high-SNR conditions in terms of the expressions of the diversity gain $(G_{ d})$ and coding gain $(G_{ c})$ achieved by the system.
In Proposition \ref{eq:asymOP}, we present the asymptotic OP expressions for both selection schemes.

\begin{proposition}\label{eq:asymOP}
	In the high SNR regime, the asymptotic OP, $P^{\text{asym}}_{\text{out}}$, for the \textit{Max-Max} and \textit{Max-Sum} selection schemes can be expressed as 
	\begin{align} \label{eq:CDF223}
		P^{\text{asym}}_{\text{out}}&=(G_{c}\bar{\gamma})^{-G_{ d}}.
	\end{align}
	Here the values of $G_{ d}$ and $G_{ c}$ are given for the \textit{Max-Max} and \textit{Max-Sum} selection schemes in Table \ref{values}.
	\begin{table}[!t]
		\caption{Values of \( G_d \) and \( G_c \) for \textit{Max-Max} and \textit{Max-Sum} Selection Schemes.}
		\label{values}
		\centering
		\begin{tabular}{|c|c|c|}
			\hline
			\textbf{Selection Scheme} & \( G_d \) & \( G_c \) \\ \hline
			\textit{Max-Max}                  & \( \frac{KN a}{2} \)        & \( \frac{\left(a \Gamma(a)\right)^{\frac{2}{a}} b^2}{L^{2}_{1}L^{2}_{2} \gamma_{\text{th}}} \) \\ \hline
			\textit{Max-Sum}                  & \( \frac{KN a}{2} \)        & \( \frac{\left(K a \Gamma(K a)\right)^{\frac{2}{Ka}} b^2}{K L^{2}_{1}L^{2}_{2} \gamma_{\text{th}}} \) \\ \hline
		\end{tabular}
	\end{table}
\end{proposition}
\begin{proof}
	It is worth mentioning that when $x \to 0$, the lower incomplete Gamma function, $\Upsilon(s,x)$, can be approximated as
	\begin{align}\label{Asym1}
		\Upsilon(s,x) \simeq \frac{x^s}{s}.
	\end{align}
	From \eqref{eq:CDF2+iidd}, at high SNR $\left(\frac{\gamma}{\bar{\gamma}} \to 0\right)$, the argument of the lower incomplete Gamma function becomes
	\begin{align}\label{Asym2}
		x = \sqrt{\frac{L^{2}_{1}L^{2}_{2}}{b^2} \frac{\gamma}{\bar{\gamma}}}.
	\end{align}
	Using \eqref{Asym1} and \eqref{Asym2} in \eqref{eq:CDF2+iidd}, then the asymptotic $F_{\gamma_{\textit{Max-Max}}}(\gamma)$, $F^{\text{asym}}_{\gamma_{\text{Max-Max}}}(\gamma)$ can be written as
	\begin{align}\label{Asym4}
		F^{\text{asym}}_{\gamma_{\textit{Max-Max}}}(\gamma) &=\left(\frac{\left(\sqrt{\frac{L^{2}_{1}L^{2}_{2}}{b^2} \frac{\gamma}{\bar{\gamma}}}\right)^{a}}{a\Gamma(a)}\right)^{KN}\nonumber\\
		&=\frac{\left(\frac{L^{2}_{1}L^{2}_{2}}{b^2} \right)^{\frac{KNa}{2}}\gamma^{\frac{KNa}{2}}}{\left(a\Gamma(a)\right)^{KN}}\left(\bar{\gamma}\right)^{-\frac{KNa}{2}}\nonumber\\
		&=\left(\frac{\left(a\Gamma(a)\right)^{\frac{2}{a}}b^2}{L^{2}_{1}L^{2}_{2} \gamma}\bar{\gamma}\right)^{-\frac{KNa}{2}}.
	\end{align}
By substituting $\gamma_{\text{th}}$ into \eqref{Asym4}, we can obtain the values of $G_{d}$ and $G_{c}$ for the \textit{Max-Max} selection scheme as in Table \ref{values}. Likewise, for the \textit{Max-Sum} selection scheme, by using \eqref{Asym1} and \eqref{Asym2} in \eqref{max-sum-iid}, then the asymptotic $F_{\gamma_{\textit{Max-Sum}}}(\gamma)$, 
$F^{\text{asym}}_{\gamma_{\textit{Max-Sum}}}(\gamma)$ can be expressed as
	\begin{align}\label{Asym3}
		F^{\text{asym}}_{\gamma_{\textit{Max-Sum}}}(\gamma) &=\left(\frac{\left(\sqrt{K\frac{L^{2}_{1}L^{2}_{2}}{b^2} \frac{\gamma}{\bar{\gamma}}}\right)^{Ka}}{Ka\Gamma(Ka)}\right)^{N}\nonumber\\
		&=\frac{\left(K\frac{L^{2}_{1}L^{2}_{2}}{b^2} \right)^{\frac{KNa}{2}}\gamma^{\frac{KNa}{2}}}{\left(Ka\Gamma(Ka)\right)^{N}}\left(\bar{\gamma}\right)^{-\frac{KNa}{2}}\nonumber\\
		&=\left(\frac{\left(Ka\Gamma(Ka)\right)^{\frac{2}{Ka}}b^2}{K L^{2}_{1}L^{2}_{2} \gamma}\bar{\gamma}\right)^{-\frac{KNa}{2}}.
	\end{align}
	By substituting $\gamma_{\text{th}}$ into \eqref{Asym3}, we can obtain the values of $G_{d}$ and $G_{c}$ for the \textit{Max-Sum} selection scheme as in Table \ref{values}, and the proof ends.
\end{proof}
\begin{remark}
From Table \ref{values}, the values of $G_d$ are identical for both schemes, while the expressions for $G_c$ differ. The $G_c $ for the \textit{Max-Sum} selection scheme depends on the number of ports $K$, whereas the \textit{Max-Max} selection scheme does not. This difference arises because, in the \textit{Max-Sum} selection scheme, all $K$ ports are activated, whereas only one port is activated in the \textit{Max-Max} selection scheme.
\end{remark}

\subsection{DOR Analysis}
The DOR is characterized as the probability that the transmission delay for a given data volume, $R$, over a wireless channel with bandwidth $B$ exceeds a specified threshold $T_\text{th}$.
%
Mathematically, this is expressed as
\begin{equation}\label{dor1}
 P_{dor} = \mathbb{P}(T_\text{dt} > T_\text{th}), 
\end{equation}
where $T_\text{dt}$ denotes the delivery time and can be written as
\begin{equation}\label{dor2}
 T_\text{dt} = \frac{R}{B \log_2(1 + \gamma)}. 
\end{equation}
The DOR corresponding to the proposed system model is provided in Proposition \ref{eq:dor}.
\begin{proposition}\label{eq:dor}
	The DORs for the \textit{Max-Max} and \textit{Max-Sum} selection schemes are given, respectively, by
	\begin{align} \label{eq:dor_max_max}
		P^{\textit{Max-Max}}_{dor}&=F_{\gamma_{\textit{Max-Max}}}\left(2^{\frac{R}{B T_\text{th}}} - 1\right),
	\end{align}
and
	\begin{align} \label{eq:dor_max_sum}
		P^{\textit{Max-Sum}}_{dor} &=F_{\gamma_{\textit{Max-Sum}}}\left(2^{\frac{R}{B T_\text{th}}} - 1\right).
	\end{align}
\end{proposition}
\begin{proof}
For $i \in \left\{\textit{Max-Max}, \textit{Max-Sum}\right\}$, the substitution of \eqref{dor2} into \eqref{dor1} yields
\begin{align}\label{dor3}
    P_{dor} &= \mathbb{P}\left( \frac{R}{B \log_2(1 + \gamma_{i})} > T_\text{th}\right)\nonumber\\
		&=\mathbb{P}\left( \frac{R}{B T_\text{th}} > \log_2(1 + \gamma_{i})\right)\nonumber\\
		&=\mathbb{P}\left( 2^{\frac{R}{B T_\text{th}}} > 1 + \gamma_{i}\right)\nonumber\\
		&=\mathbb{P}\left( \gamma_{i} < 2^{\frac{R}{B T_\text{th}}} - 1\right)\nonumber\\
		&=F_{\gamma_{i}}\left(2^{\frac{R}{B T_\text{th}}} - 1\right),
\end{align}
By substituting Propositions \ref{proposition_CDF_Max-Max} and \ref{proposition_CDF_Max-Sum} into \eqref{dor3}, we obtain Proposition \ref{eq:dor}, and the proof ends.
\end{proof}



\begin{figure}[t]
\centering
{\includegraphics[width=8.8cm,height=5.55cm]{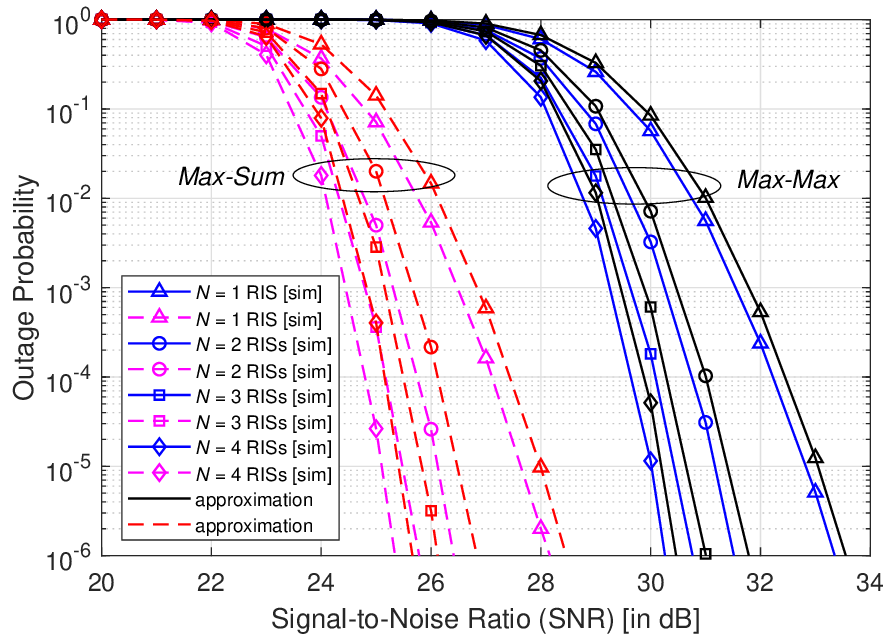}}\\
\caption{OP versus SNR with varying the number of RISs, $N$, when $M$ = 16 elements, $m$ = 1, $K$ = 4 ports, and $A = 1\lambda^2$.}
\label{Figure2}
\end{figure}



\section{Numerical Analysis} \label{Numerical Analysis}

In this section, the performance of the proposed multi-RIS-assisted FAS is evaluated in terms of OP and DOR by varying key parameters including the number of RISs, $N$, and their reflecting elements, $M$, number of ports, $K$, antenna size, $\lambda d_{1} \times \lambda d_{2}$, the shape parameter, $m_{k,n,i}$, the amount of data, $R$, and the bandwidth, $B$. In addition, we set the following simulation parameters, $\lambda  = 1$ m, $\gamma_\text{th} = 0$ dB, $T_\text{th} = 3$ ms $L_{n,1} = L_{n,2} = 20$ m, $\Omega = \Omega_{k,n,i} = 1$ for all figures. 


Fig. \ref{Figure2} shows the OP performance of multi-RIS-assisted FAS versus SNR. In this figure, the OP performance is plotted for various numbers of RISs $(N\in\{1,2,3,4\})$ and equal Nakagami parameters $(m = m_{k,n,i}=1)$, which is a special case of Rayleigh fading. Additionally, the proposed \textit{Max-Max} and \textit{Max-Sum} selection schemes are compared, and the single-RIS-assisted FAS with $N= 1$ is also evaluated for the sake of fair comparison. As seen in Fig. \ref{Figure2}, the approximated analytical results are validated by the simulations and are near the exact ones. Moreover, the OP performance improves as the number of RISs increases. For example, to achieve an OP of $10^{-5}$ under \textit{Max-Max} selection scheme, deploying multiple RISs yields SNR gains of approximately $2$ dB, $2.5$ dB and $3$ dB for $N = 2, 3$ and $4$, respectively, compared to the single-RIS case. In addition, the \textit{Max-Sum} selection scheme outperforms the \textit{Max-Max} selection scheme because it exploits more spatial paths. For instance, in the case of $N = 4$, to achieve an OP of $10^{-5}$, there is about a $5$ dB SNR gain for the \textit{Max-Sum} selection scheme over the \textit{Max-Max} scheme. 
Moreover, the same SNR gain of about $5$ dB is achieved in the former scheme over the latter scheme regardless of the number of RISs.

Figs. \ref{Figure3} and \ref{Figure4} illustrate the OP performance versus SNR for different numbers of FAS ports, respectively, with a fixed FA size $(A = 1\lambda^2)$ and variable FA sizes $(A\in\{1\lambda^2, 2\lambda^2, 3\lambda^2, 4\lambda^2\})$\footnote{In this paper, the cases, $K = 2 \times 2$, $2 \times 3$, $2 \times 4$, $2 \times 5$, and $4 \times 4$ with $A = 1\lambda \times 1\lambda$, $1\lambda \times 2\lambda$, $1\lambda \times 2\lambda$, $1\lambda \times 3\lambda$, and $2\lambda \times 2\lambda$, are considered.}. 
%
%
It is evident from Fig. \ref{Figure3} that both \textit{Max-Max} and \textit{Max-Sum} selection schemes exhibit identical performance for the case of a single port, i.e., $K = 1$, as the single port will always be selected in both schemes. 
It is also clearly observed from Figs. \ref{Figure3} and \ref{Figure4} that the OP performance enhances, respectively, as the number of ports and FA size increase. This is mainly because increasing the number of ports introduces a stronger spatial correlation.
For example, from Fig. \ref{Figure3}, to achieve an OP of $10^{-4}$ for the proposed \textit{Max-Max} selection scheme, there are about $2.7$ dB, $3.7$ dB, $4.3$ dB, and $4.7$ dB SNR gains for $K = 4, 9, 16$ and $25$ over the case with $K = 1$, respectively. On the other hand, to achieve an OP of $10^{-4}$ for the \textit{Max-Sum} selection scheme, there are almost $7.3$ dB, $11.2$ dB, $14$ dB, and $16.1$ dB SNR gains for $K = 4, 9, 16$ and $25$ over the case with $K = 1$, respectively. 
Similarly, Fig. \ref{Figure4} shows that increasing the FA size results in increased spatial correlation and thereby improved OP performance.
These two figures show that the proposed \textit{Max-Sum} selection scheme achieves higher SNR gains as the number of ports and the size of FA increase, as it exploits more spatial paths, thereby enhancing the signal diversity and overall system performance.

\begin{figure}[t]
\centering
{\includegraphics[width=8.5cm,height=5.35cm]{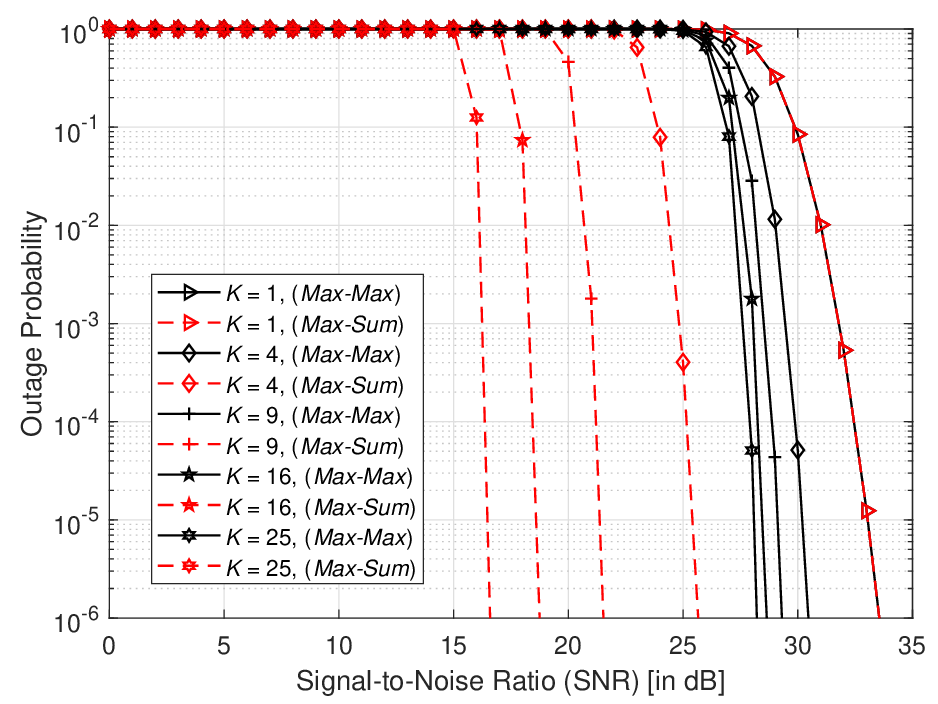}}
\caption{OP versus SNR with varying the number of ports, $K$, when $N$ = 4 RISs, $M$ = 16 elements, $m$ = 1, and $A = 1\lambda^2$.}
\label{Figure3}
\end{figure}

\begin{figure}[t!]
\centering
{\includegraphics[width=8.5cm,height=5.35cm]{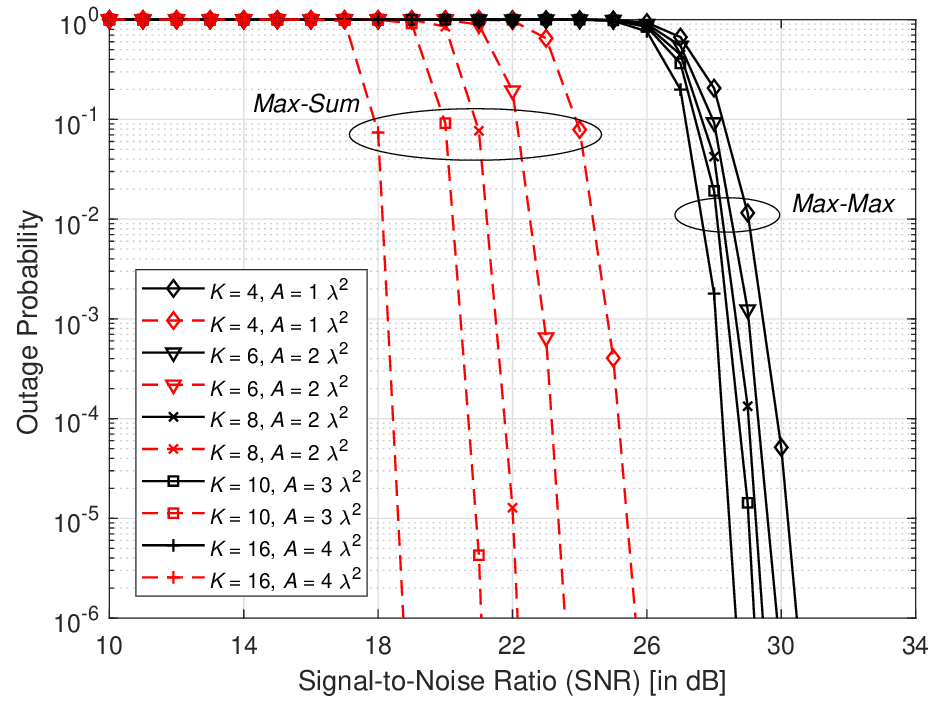}}
\caption{OP versus SNR with varying the FA size, ${\it{A}} = d_{1} \lambda \times d_{2} \lambda$, and varying the number of ports, $K$, when $M$ = 16 elements, $m$ = 1, and $N$ = 4 RISs.}
\label{Figure4}
\end{figure}


The OP performance of multi-RIS-assisted FAS for both \textit{Max-Max} and \textit{Max-Sum} selection schemes is presented in Fig. \ref{Figure5} for various numbers of reflecting RIS elements $(M\in\{8, 16, 32, 64\})$. It is observed from this figure that the OP performance gets better as the number of RIS elements increases. For instance, to achieve an OP of $10^{-5}$ for both selection schemes, there is approximately a $5$ dB SNR gain by doubling the number of reflecting RIS elements.

In order to investigate the impact of the channel conditions on the OP performance of both 
\textit{Max-Max} and \textit{Max-Sum} 
selection schemes, OP versus SNR is evaluated in Fig. \ref{Figure6} with different fading parameters, $(m=m_{k,n,i}\in\{1, 2, 3\})$. This comparison shows that the OP performance improves as the channel condition improves, i.e., the values of $m$ and $m_{k,n,i}$ increase.

\begin{figure}[t]
\centering
{\includegraphics[width=8.5cm,height=5.4cm]{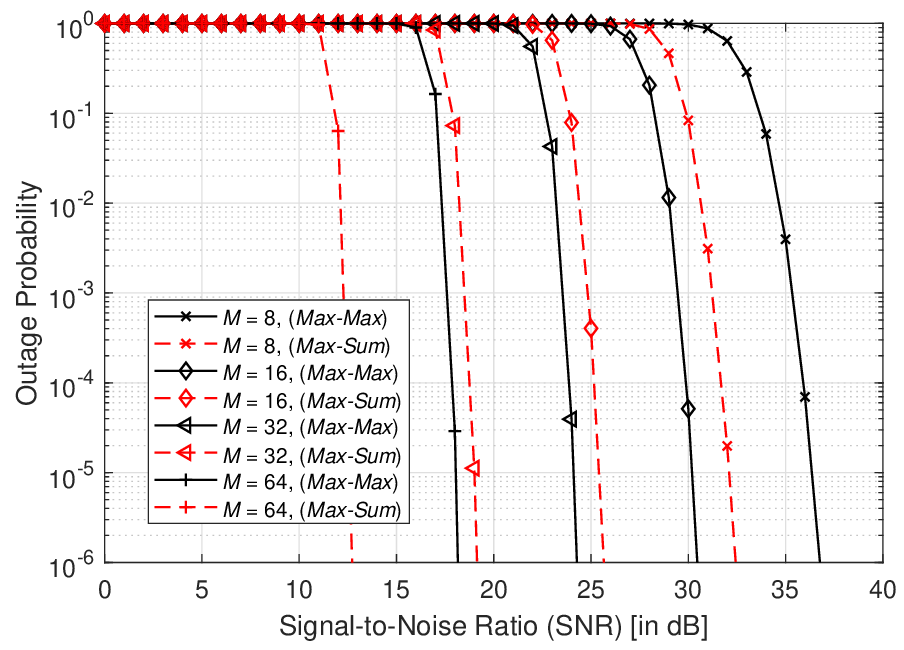}}
\caption{OP versus SNR with varying the number of reflecting RIS elements, $M$, when $K$ = 4 ports, $m$ = 1, $A = 1\lambda^2$, and $N$ = 4 RISs.}
\label{Figure5}
\end{figure}

\begin{figure}[t!]
\centering
{\includegraphics[width=8.5cm,height=5.4cm]{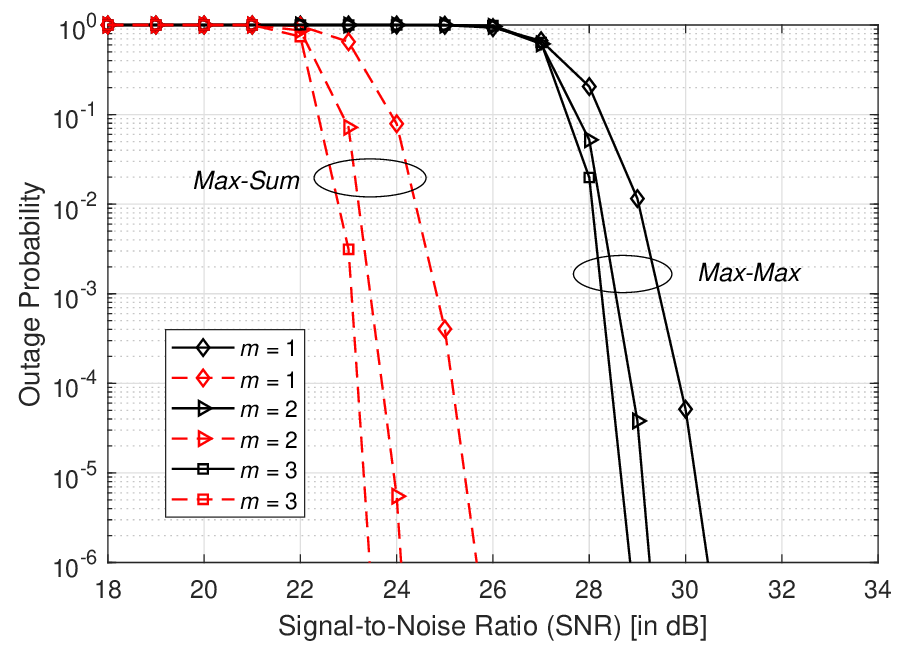}}
\caption{OP versus SNR with varying $m$, when $M$ = 16 elements, $K$ = 4 ports, $m$ = 1, $A = 1\lambda^2$, and $N$ = 4 RISs.}
\label{Figure6}
\end{figure}

Finally, Figs. \ref{Figure7a} and \ref{Figure7b} illustrate the DOR performance versus SNR when varying the amount of data while keeping the bandwidth fixed, and when varying the bandwidth with keeping the amount of data fixed, respectively. It is clear from Fig. \ref{Figure7a} that DOR performance degrades as the amount of transmitted data, $R$, increases under a fixed bandwidth of $B =$ 1 MHz due to the longer transmission time required. This makes it challenging to meet the latency constraint with high data rates, resulting in a significant increase in the delay outage. In addition, Fig. \ref{Figure7b} illustrates that DOR performance improves as the channel bandwidth increases. This shows that a preset amount of data can be transmitted with a reduced delay as the bandwidth increases.

\begin{figure*}[t]
\centering
\subfloat[\label{Figure7a}]{%
\includegraphics[width=8.5cm,height=5.5cm]{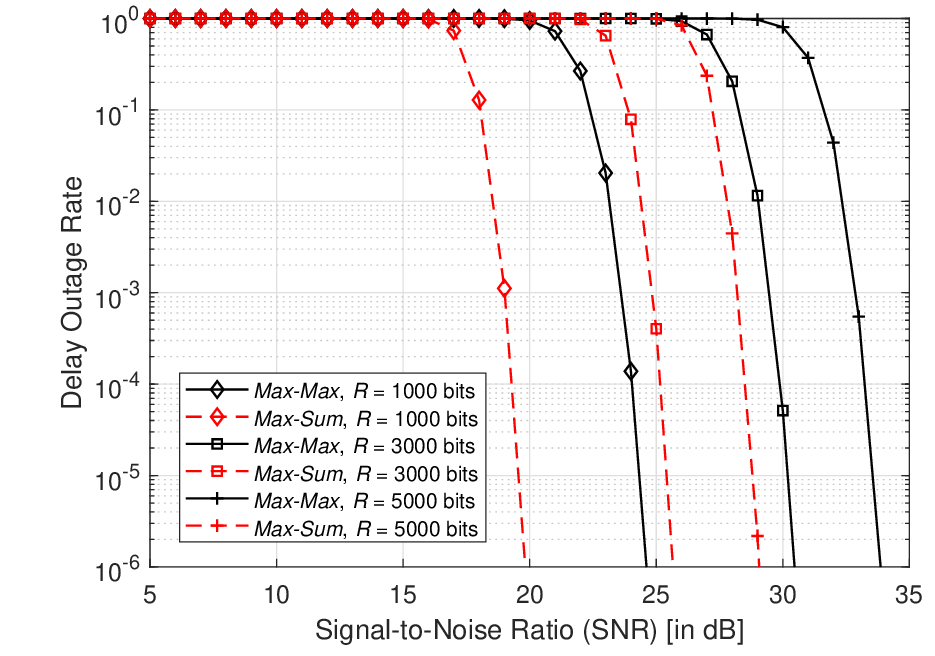}}
\hspace{\fill}
\subfloat[\label{Figure7b}]{%
\includegraphics[width=8.5cm,height=5.5cm]{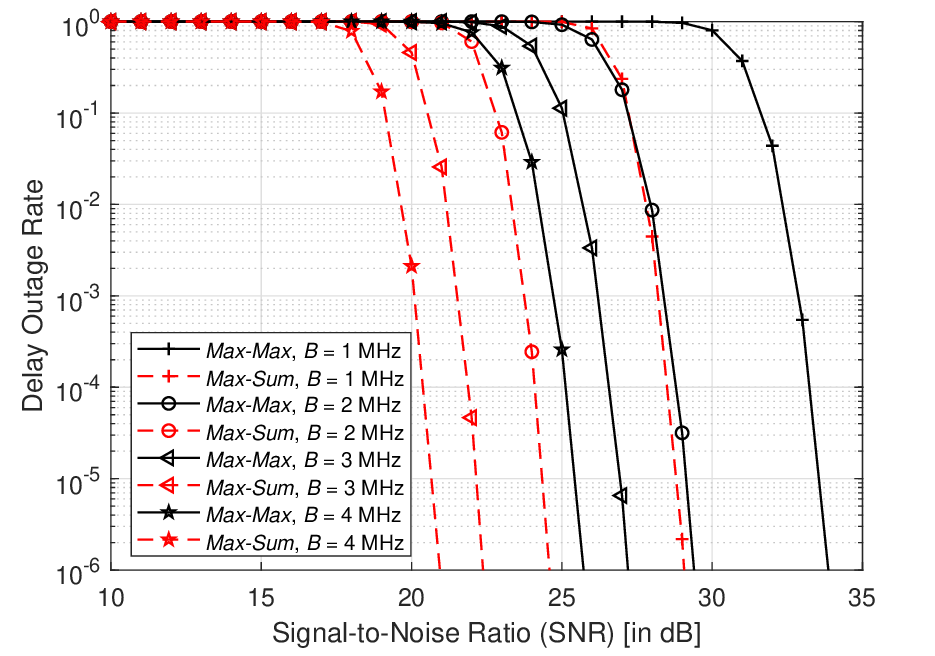}}\\
\caption{DOR versus SNR when $N$ = 4 RISs, $M$ = 16 elements, $K$ = 4 ports, $m$ = 1, and $A = 1\lambda^2$, (a) with varying the amount of data, $R$, for a fixed $B$ = 1 MHz, (b) with varying the bandwidth, $B$, for a fixed $R$ = 3000 bits.}
\label{Figure7}
\end{figure*}


\section{Conclusion} \label{Conclusion}
This paper proposed two novel selection schemes, referred to as \textit{Max-Max} and \textit{Max-Sum}, to enhance the system performance of the multi-RIS-assisted FAS. We derived the CDFs of the SNR for both selection schemes and used them to obtain approximate theoretical OP and DOR expressions. In order to have further insights into the system performance in terms of diversity and coding gains, we carried out the OP asymptotic analysis. We finally verified the analytical results using Monte Carlo simulations. The simulation results show that the OP performance improves as the number of RISs, the number of ports, FA size, and the number of reflecting RIS elements increase and as the channel conditions improve. In addition, the proposed \textit{Max-Sum} selection scheme is superior to the \textit{Max-Max} selection scheme. The results also demonstrate that the SNR gain achieved by the \textit{Max-Sum} selection scheme improves as the number of ports and the size  of the FA increase. Moreover, there is a degradation in the DOR performance once the amount of transmitted data increases, while there is an improvement when the channel bandwidth becomes larger.

\bibliographystyle{IEEEtran}

\bibliography{Ref}

\begin{thebibliography}{10}
\providecommand{\url}[1]{#1}
\csname url@samestyle\endcsname
\providecommand{\newblock}{\relax}
\providecommand{\bibinfo}[2]{#2}
\providecommand{\BIBentrySTDinterwordspacing}{\spaceskip=0pt\relax}
\providecommand{\BIBentryALTinterwordstretchfactor}{4}
\providecommand{\BIBentryALTinterwordspacing}{\spaceskip=\fontdimen2\font plus
\BIBentryALTinterwordstretchfactor\fontdimen3\font minus
  \fontdimen4\font\relax}
\providecommand{\BIBforeignlanguage}[2]{{%
\expandafter\ifx\csname l@#1\endcsname\relax
\typeout{** WARNING: IEEEtran.bst: No hyphenation pattern has been}%
\typeout{** loaded for the language `#1'. Using the pattern for}%
\typeout{** the default language instead.}%
\else
\language=\csname l@#1\endcsname
\fi
#2}}
\providecommand{\BIBdecl}{\relax}
\BIBdecl

\bibitem{6G}
H.~Tataria, M.~Shafi, A.~F. Molisch, M.~Dohler, H.~Sjöland, and F.~Tufvesson,
  ``{6G} wireless systems: Vision, requirements, challenges, insights, and
  opportunities,'' \emph{Proceedings of the IEEE}, vol. 109, no.~7, pp.
  1166--1199, 2021.

\bibitem{MIMO}
E.~Björnson, J.~Hoydis, and L.~Sanguinetti, ``Massive {MIMO} has unlimited
  capacity,'' \emph{IEEE Trans. Wireless Commun.}, vol.~17, no.~1, pp.
  574--590, 2018.

\bibitem{8804165}
M.~A. Albreem, M.~Juntti, and S.~Shahabuddin, ``Massive {MIMO} detection
  techniques: {A} survey,'' \emph{IEEE Communications Surveys \& Tutorials},
  vol.~21, no.~4, pp. 3109--3132, 2019.

\bibitem{RIS}
M.~Di~Renzo, A.~Zappone, M.~Debbah, M.-S. Alouini, C.~Yuen, J.~de~Rosny, and
  S.~Tretyakov, ``Smart radio environments empowered by reconfigurable
  intelligent surfaces: How it works, state of research, and the road ahead,''
  \emph{IEEE J. Sel. Areas Commun.}, vol.~38, no.~11, pp. 2450--2525, 2020.

\bibitem{FAS}
W.~K. New, K.-K. Wong, H.~Xu, C.~Wang, F.~R. Ghadi, J.~Zhang, J.~Rao, R.~Murch,
  P.~Ramírez-Espinosa, D.~Morales-Jimenez, C.-B. Chae, and K.-F. Tong, ``A
  tutorial on fluid antenna system for {6G} networks: Encompassing
  communication theory, optimization methods and hardware designs,'' \emph{IEEE Commun. Surveys \& Tut.}, pp. 1--1, 2024.

\bibitem{RIS_idea}
E.~Basar, M.~Di~Renzo, J.~De~Rosny, M.~Debbah, M.-S. Alouini, and R.~Zhang,
  ``Wireless communications through reconfigurable intelligent surfaces,''
  \emph{IEEE Access}, vol.~7, pp. 116\,753--116\,773, 2019.

\bibitem{FAS_idea}
A.~F. M.~S. Shah, M.~Ali~Karabulut, E.~Cinar, and K.~M. Rabie, ``A survey on
  fluid antenna multiple access for {6G}: A new multiple access technology that
  provides great diversity in a small space,'' \emph{IEEE Access}, vol.~12, pp.
  88\,410--88\,425, 2024.

\bibitem{RIS_FA_1}
K.-K. Wong, K.-F. Tong, and C.-B. Chae, ``Fluid antenna system—part {III}: A
  new paradigm of distributed artificial scattering surfaces for massive
  connectivity,'' \emph{IEEE Commun. Lett.}, vol.~27, no.~8, pp. 1929--1933,
  2023.

\bibitem{9715064}
Z.~Chai, K.-K. Wong, K.-F. Tong, Y.~Chen, and Y.~Zhang, ``Port selection for
  fluid antenna systems,'' \emph{IEEE Commun. Lett.}, vol.~26, no.~5, pp.
  1180--1184, 2022.

\bibitem{10103838}
M.~Khammassi, A.~Kammoun, and M.-S. Alouini, ``A new analytical approximation
  of the fluid antenna system channel,'' \emph{IEEE Trans. Wireless Commun.},
  vol.~22, no.~12, pp. 8843--8858, 2023.

\bibitem{10253941}
F.~Rostami~Ghadi, K.-K. Wong, F.~J. López-Martínez, and K.-F. Tong,
  ``Copula-based performance analysis for fluid antenna systems under arbitrary
  fading channels,'' \emph{IEEE Commun. Lett.}, vol.~27, no.~11, pp.
  3068--3072, 2023.

\bibitem{10678877}
F.~Rostami~Ghadi, K.-K. Wong, F.~Javier López-Martínez, C.-B. Chae, K.-F.
  Tong, and Y.~Zhang, ``A {Gaussian} copula approach to the performance
  analysis of fluid antenna systems,'' \emph{IEEE Trans. Wireless Commun.},
  vol.~23, no.~11, pp. 17\,573--17\,585, 2024.

\bibitem{10319727}
Y.~Hou, J.~Mao, S.~Zhang, Q.~Cui, X.~Tao, W.~Li, and Y.~Chen, ``A copula-based
  approach to performance analysis of fluid antenna system with multiple fixed
  transmit antennas,'' \emph{IEEE Wireless Commun. Lett.}, vol.~13, no.~2, pp.
  501--504, 2024.

\bibitem{10623405}
P.~Ramírez-Espinosa, D.~Morales-Jimenez, and K.-K. Wong, ``A new spatial
  block-correlation model for fluid antenna systems,'' \emph{IEEE Trans.
  Wireless Commun.}, vol.~23, no.~11, pp. 15\,829--15\,843, 2024.

\bibitem{10375698}
X.~Lai, T.~Wu, J.~Yao, C.~Pan, M.~Elkashlan, and K.-K. Wong, ``On performance
  of fluid antenna system using maximum ratio combining,'' \emph{IEEE Commun.
  Lett.}, vol.~28, no.~2, pp. 402--406, 2024.

\bibitem{RIS_FA_2}
J.~Chen, Y.~Xiao, J.~Zhu, Z.~Peng, X.~Lei, and P.~Xiao, ``Low-complexity
  beamforming design for {RIS}-assisted fluid antenna systems,'' in \emph{2024
  IEEE Int. Conf. Commun. Workshops}, 2024, pp. 1377--1382.

\bibitem{RIS_FA_3}
J.~Zhu, Q.~Luo, G.~Chen, P.~Xiao, Y.~Xiao, and K.-K. Wong, ``Fluid antenna
  empowered index modulation for {RIS}-aided {mmWave} transmissions,''
  \emph{IEEE Trans. Wireless Commun.}, pp. 1--1, 2024.

\bibitem{RIS_FA_4}
X.~Lai, J.~Yao, K.~Zhi, T.~Wu, D.~Morales-Jimenez, and K.-K. Wong, ``{FAS-RIS}:
  A block-correlation model analysis,'' \emph{IEEE Trans. Veh. Technol.}, pp.
  1--6, 2024.

\bibitem{RIS_FA_5}
F.~Rostami~Ghadi, K.-K. Wong, W.~K. New, H.~Xu, R.~Murch, and Y.~Zhang, ``On
  performance of {RIS}-aided fluid antenna systems,'' \emph{IEEE Wireless
  Commun. Lett.}, vol.~13, no.~8, pp. 2175--2179, 2024.

\bibitem{10826703}
G.~Wang, J.~Zhang, K.~Xu, Y.~Ni, and Y.~Wu, ``Performance analysis of
  multi-{RIS}-aided fluid antenna systems under {Nakagami}-$m$ fading
  channels,'' in \emph{2024 16th Int. Conf. Wireless Commun. Sig. Process.
  (WCSP)}, 2024, pp. 360--365.

\bibitem{10858773}
J.~Yao, J.~Zheng, T.~Wu, M.~Jin, C.~Yuen, K.-K. Wong, and F.~Adachi,
  ``{FAS-RIS} communication: Model, analysis, and optimization,'' \emph{IEEE
  Trans. Veh. Technol.}, pp. 1--6, 2025.

\bibitem{CDF_theory}
M.~Aldababsa, A.~M. Salhab, A.~A. Nasir, M.~H. Samuh, and D.~B. da~Costa,
  ``Multiple {RISs}-aided networks: Performance analysis and optimization,''
  \emph{IEEE Trans. Veh. Technol.}, vol.~72, no.~6, pp. 7545--7559, 2023.

\bibitem{Grad}
I.~S. Gradshteyn and I.~M. Ryzhik, \emph{Tables of Integrals, Series and
  Products}, 6th~ed.\hskip 1em plus 0.5em minus 0.4em\relax San Diego: Academic
  Press, 2000.

\end{thebibliography}

\end{document}